\documentclass[]{interact}
\usepackage{float}
\usepackage{epstopdf}
\usepackage[caption=false]{subfig}

\usepackage{hyperref}
\usepackage[natbibapa,nodoi]{apacite}
\theoremstyle{plain}
\newtheorem{theorem}{Theorem}[section]
\newtheorem{lemma}[theorem]{Lemma}
\newtheorem{corollary}[theorem]{Corollary}
\newtheorem{proposition}[theorem]{Proposition}
\newtheorem{problem}[theorem]{Problem}

\theoremstyle{definition}

\theoremstyle{remark}
\newtheorem{remark}{Remark}

\begin{document}


\title{Least-Squares Parameter Estimation for State-Space Models with State Equality Constraints}

\author{
\name{Rodrigo~A.~Ricco\textsuperscript{a}\thanks{CONTACT Rodrigo~A.~Ricco. Email: ricco@deelt.ufop.br} and Bruno~O.~S.~Teixeira\textsuperscript{b}}
\affil{\textsuperscript{a}Universidade Federal de Ouro Preto, Department of Electrical Engineering, Minas Gerais, Brazil; \textsuperscript{b}Universidade Federal de Minas Gerais, Department of Electronic Engineering, Minas Gerais, Brazil}
}
\maketitle
\begin{abstract}
If a dynamic system has active constraints on the state vector and they are known, then
taking them into account during modeling is often advantageous. Unfortunately, in the constrained discrete-time state-space estimation, the state equality constraint is defined for a parameter matrix and not on a parameter vector as commonly found in regression problems.
To address this problem, firstly, we show how to rewrite the state equality constraints as equality constraints on the state matrices to be estimated.
Then, we vectorize the matricial least squares problem defined for modeling state-space systems such that  any method from the equality-constrained least squares framework may be employed. Both time-invariant and time-varying cases are considered as well as the case where the state equality constraint is not exactly known.
\end{abstract}


\section{Introduction}
In some dynamic systems, dynamics evolve  with variables satisfying inequality or equality constraints \citep{2005_Goodwin_Seron_Dona}. 
For example, the species concentrations are non-negative in chemical reactions \citep{1995_Massicotte_Morawski_Barwicz}.  Likewise,  in the quaternion-based attitude representation, the attitude vector must have unitary norm \citep{2003_Crassidis_Markley}; and, for  ground vehicle tracking problems, the road networks can be viewed as equality constraints on the trajectory \citep{Xu2013}. 
The combination of single or multiple cells  in  biological processes may be represented as a compartment with constant volume \citep{1974_Mohler}. In addition, compartmental models have applications in classical circuit models, structural models and complex networks, among others, as pointed out in \citet{1993_Hyland_Bernstein}.

In this work, we are specifically concerned with linear state-space dynamic systems satisfying linear equality constraints on the state vector. The scenario we have in mind is the one in which we have dynamical data collected from the dynamic system as well as auxiliary information (written as equality constraints on the state vector) from first principles. Consider the following examples:  localization of a land vehicle for which the road map represents a constraint on the trajectory \citep{Xu_2017}; the flight formation of two targets, where the distance between the targets is constant \citep{Xu2013}; the monitoring of the nitrogen flow in a tropical forest where the amount of nitrogen is constant \citep{Walter1999}; the experiment of the wet granulation of lactose with deionized water carried out in a ploughshare mixer with constant volume \citep{Lee2017}; and an interconnected tank system for which prior information on the total volume is available \citep{Holzel2014}.

At this point, one may argue that variable reduction \citep{Holzel2014,2016_Li} may be employed to avoid enforcing the equality constraint on the state vector. However, this approach yields a reduced state vector with a different physical meaning, which is not desirable in many applications. Moreover, if the equality constraint is time varying, keeping a constant state vector parametrization is of interest.

In the last years, the problem of state estimation for both linear and nonlinear equality-constrained dynamic systems has received great attention from the community \citep{Babacan2008,Teixeira2009,2010_Simon,Teixeira2008,Xu2013,Rengaswamy2013,Duan_and_Li2015,Xu_2017}. The problem of modeling such systems is less often addressed \citep{Xu2013,2016_Li,Xu_2017}.  
In the latter works, a two-step modeling procedure is employed. An unconstrained model (auxiliary dynamics) is first obtained and, by projection,  such model is fused with the state equality constraint.

Parameter estimation with known equality linear constraints is a solved problem.
Auxiliary information such as  static function, static gain, and fixed-point location, can be written in the form of linear equality constraints on the parameters of NARX (Nonlinear autoregressive with exogenous inputs) polynomial and RBF (Radial basis function) network models, for instance \citep{2011_Teixeira_Aguirre,2007_Aguirre_Alves_Correa}.
If the dynamics are time invariant, one may use the {\it batch} equality constrained least squares \citep{Bjork1996, draper1998}.  For problems in which auxiliary information is uncertain, the compromise between prediction performance and the equality constraint satisfaction is treated by means a tuning parameter in \citet{2011_Teixeira_Aguirre}.
In \citet{arablouei2015} the relaxed solution of the {\it batch} equality constrained least squares is addressed to solve the same problem. For a {\it recursive} solution, it suffices to use the classical recursive least squares with a proper initialization as shown in \citet{Zhou_et_al_2001, ZhuLi2007}. 
For convenience, in this work, we present  this result using a different perspective in Proposition \ref{prop:rcls}.
However, for time-varying systems, the equality parameter constraint must be enforced at every time instant in the recursive least squares equations \citep{Alenany2013}. In this regard \cite{VINCENT_2018}, exploring connections between Kalman filter and least squares, enforce equality constraint on the Kalman gain \citep{Teixeira2008} in order to guarantee that the estimator is unbiased. Conversely, in this manuscript, we enforce constraints on the model matrices in order to guarantee a model whose state vector satisfy an equality constraint.

If we assume that all state components are directly measured, then least squares methods may be used to estimate the matrices of the linear state-space model. Otherwise, subspace methods must be used \citep{Trnka2009,Alenany2011,privara2012,Alenany2013,2018_Youqing} with least squares as a possible internal step.
Consider the case of fully measured state vector, for which auxiliary information on the state vector is known in the form of an equality constraint. 
How to estimate the state matrices of the system such that the free-run prediction of its state vector satisfies the known equality constraint? To address this question (Problem \ref{problem1}), one must first be able to mathematically map the equality constraint on the state vector onto an equality constraint on the state space matrices (parameters) to be estimated.
In this manuscript, this problem is solved for time-varying linear dynamic systems by adapting a result from \citet{Teixeira2009}; see Lemmas \ref{lemma:ec} and \ref{lemma:ec2}. However, the aforementioned equality-constrained least-squares methods cannot be used to enforce such equality constraints because, in this case, the constraint is defined for a parameter matrix and not on a parameter vector as commonly found in regression problems; see Remark \ref{rem:mi}. To circumvent this problem, we vectorize the matricial least squares problem defined for modeling state-space systems using the vectorization operator and Kronecker product as in \citet{privara2012} such that the existing equality constrained least squares framework may be employed; see Proposition \ref{prop:cls}. The contributions of this manuscript are: (i) to address the two aforementioned problems as a single mathematical problem and to solve such state-space modeling problem with state equality constraints, and (ii)  to explore the connections among the papers that address similar problems in the literature \citep{2011_Teixeira_Aguirre,arablouei2015,Zhou_et_al_2001,ZhuLi2007,Alenany2013}.
Here both time-invariant and time-varying cases are considered. Finally, as in \citet{2011_Teixeira_Aguirre}, the case in which the auxiliary information is uncertain is also addressed.

This  document is organized as follows. Section \ref{sec:ps} formulates the gray-box system identification problem under investigation. In Section \ref{sec:background} we review known equality-constrained parameter estimation methods for both time-invariant and time-varying systems. Section \ref{sec:CLS_LEC} solves the problem formulated in Section \ref{sec:ps}, presenting the main contributions of this manuscript. In Section \ref{sec:SR}, numerical examples illustrate the applicability of the proposed approaches. Finally, in Section \ref{sec:conclusions}, the concluding remarks are discussed.

Notation is set as follows in this manuscript. $I_n$ and $0_{m\times n}$ stand respectively for the $n$-dimensional identity matrix and $m\times n$-dimensional zero matrix. $\otimes$ is the Kronecker product and ${\rm vec}$ is the vectorizer operator. 

\section{Problem Statement}
\label{sec:ps}
Consider the linear discrete-time  state-space system
\begin{eqnarray}
 x_{k+1}&=&Ax_{k} + Bu_{k} + Gw_{k}, \label{eq:proc}\\
  y_{k}& =& x_{k} + v_{k}, \label{eq:obs}
 \end{eqnarray}
where   $x_{k} \in \mathbb{R}^{n}$ is the state vector,  $u_{k} \in \mathbb{R}^{p}$ is the input vector,  $y_k \in \mathbb{R}^m$ is the measured output vector,  $w_k \in \mathbb{R}^q$, $q < n$, is the zero-mean process noise with  covariance $Q \in \mathbb{R}^{q \times q}$ and $v_k  \in \mathbb{R}^m$ is the zero-mean measurement noise. Note that all the states are assumed to be measured. Assume that the noise terms are mutually uncorrelated.  The matrices  $A \in \mathbb{R}^{n\times n}$, $B \in \mathbb{R}^{n\times p}$ and $G \in \mathbb{R}^{n\times q}$ are not assumed to be known. 
Assume that the system \eqref{eq:proc} is asymptotically stable. 
In addition, assume that the state vector satisfies the equality constraint
\begin{equation}\label{restricoes_igualdade}
   Sx_{k}= s, 
\end{equation}
where $S \in \mathbb{R}^{n_{r} \times n}$, $s \in \mathbb{R}^{n_{r} \times 1}$ and $1 \le n_{r} \le n-q$ is the number of constraints. Without loss of generality, we assume that ${\rm rank}(S) = n_r$. 

The state-space model \eqref{eq:proc}-\eqref{eq:obs} can be rewritten as 
\begin{equation} \label{state_1}
y^{T}_{k+1}=\left[\begin{array}{cc}
         y_{k}^{T} & u^{T}_{k}
       \end{array}\right]\left[\begin{array}{c}
         A^{T} \\
         B^{T}
       \end{array}\right] +e_{k+1}^{T},
\end{equation}
where 
\begin{equation}\label{eq:noise}
e_{k+1} \triangleq v_{k+1} - Av_k+ G w_{k}.
\end{equation}
Next, assume that sequences of $u_k$ and $y_k$ are known for $k = 0,\ldots,N$, such that we have
\begin{equation}
\mathcal{Y} = \mathcal{X}\Theta + \mathcal{E}, \label{eq:ss_lr}
\end{equation}
where
\begin{eqnarray}
&\mathcal{Y}  \triangleq \left[\begin{array}{c}
         y_{1}^{T} \\
        \vdots  \\
         y_{N}^{T}
       \end{array}\right],~~ \mathcal{X} \triangleq \left[\begin{array}{cc}
            y_{0}^{T} & u^{T}_{0}\\
            \vdots & \vdots \\
          y_{N-1}^{T} & u^{T}_{N-1}
       \end{array}\right], ~~ \mathcal{E}  \triangleq \left[\begin{array}{c}
         e_{1}^{T} \\
        \vdots  \\
         e_{N}^{T}
       \end{array}\right], \nonumber 
       \end{eqnarray}
       \begin{eqnarray}
       &\Theta \triangleq \left[\begin{array}{c}
         A^{T} \\
         B^{T}
       \end{array}\right],  \label{eq:Theta}
\end{eqnarray}
where $\mathcal{{Y},~E} \in \mathbb{R}^{N \times n}$, $\mathcal{{X}} \in \mathbb{R}^{N \times (n+p)}$ and ${\Theta} \in \mathbb{R}^{ (n+p) \times n}$.

Define the least squares cost function 
\begin{equation}\label{JLS_estados_mensurados_matrizes}
J_{LS}( \hat{{{\Theta}}})\triangleq \left(\mathcal{{Y}}-\mathcal{{X}}\hat{{\Theta}} \right)^{T}\left(\mathcal{{Y}}-\mathcal{{X}}\hat{{\Theta}}\right).
\end{equation}

\begin{problem} \label{problem1}
The equality-constrained state-space modeling problem is to obtain the minimizer $\hat{{{\Theta}}}_{\rm CLS}$ of \eqref{JLS_estados_mensurados_matrizes} such that the states $x_k$ of the corresponding process model \eqref{eq:proc} satisfy the equality constraint \eqref{restricoes_igualdade}.
\end{problem}

Recall that this paper does not address the problem of state estimation, although the parameter estimation problem under investigation can be recast as a state estimation problem under proper assumptions.

\section{Background on Equality-Constrained \\Least Squares}\label{sec:background}
\subsection{Time-invariant case}
Consider the linear regression model 
\begin{equation}\label{regression_eq}
z_k =\psi^{T}_{k-1}\hat{\theta}+\xi_k,
\end{equation}
where $z_k \in \mathbb{R}$ is the measured output, $\psi_{k-1} \in \mathbb{R}^{n_p}$ is the known regressor vector, $\xi_k \in \mathbb{R}$ is the  residue, and $\hat{\theta} \in \mathbb{R}^{n_p}$ is the unknown parameter vector to be estimated.
Recall that \eqref{regression_eq} may represent the dynamic model for a linear-in-the-parameters MISO system. 
Assume that a set of $N$ observations of $z$ and $\psi$ in \eqref{regression_eq} are available such that we have
\begin{equation}\label{regression_eq2}
Z=\Psi \hat{\theta}+\Xi.
\end{equation}
Assume that $\Psi$ has full column rank such that $\Psi$ is left invertible.

Now, assume that the parameters $\hat{\theta}$ must satisfy a  set of $n_{r}$ linear equality constraints given by 
\begin{equation}\label{equality_constraint}
D\hat{\theta}=d,
\end{equation}
where $d \in \mathbb{R}^{n_{r}}$ and $D \in \mathbb{R}^{n_{r} \times n_{p}}$ with ${\rm rank}(D)=n_p$. Next, define the least squares cost function 
\begin{equation} \label{JLS_estados_mensurados}
J_{\textrm{LS}}( \hat{{\theta}})\triangleq \left(Z-\Psi \hat{{\theta}}\right)^{T}\left(Z-\Psi \hat{{\theta}}\right).
\end{equation}
 
Then, the minimizer of \eqref{JLS_estados_mensurados} subject to \eqref{equality_constraint} is given by \citet{Bjork1996,draper1998}
\begin{equation}\label{CLS_estimator2}
\hat{{\theta}}_{\rm{CLS}}= \mathcal{P}_{\mathcal{N}(D)}\hat{{\theta}}_{\rm{LS}}+  (I_{n_{p}} - \mathcal{P}_{\mathcal{N}(D)})\bar{d},
\end{equation}
where
\begin{equation}\label{LS_estimator}
\hat{{\theta}}_{\textrm{LS}} =(\Psi^{T}\Psi)^{-1}\Psi^T Z,
\end{equation}
\begin{equation} \mathcal{P}_{\mathcal{N}(D)} \triangleq I_{n_{p}}-LD, \label{eq:projector} \end{equation}
\begin{equation}\label{CLS_projector}
L\triangleq (\Psi^{T}\Psi)^{-1}D^{T}[D(\Psi^{T}\Psi)^{-1}D^{T}]^{-1},
\end{equation}
and 
\begin{equation} \bar{d} \triangleq D^T(D D^{T})^{-1}d \label{eq:offset}\end{equation}
is an offset. Therefore, the equality constraint $D\hat{\theta}_{\rm{CLS}}=d$ is exactly satisfied. The estimator \eqref{LS_estimator} is known as the classical {\em least squares} (LS) and \eqref{CLS_estimator2} is known as the {\em equality-constrained least squares} (CLS).

\begin{remark} \label{rem:rel}
Augment the matrices in \eqref{regression_eq2} and \eqref{JLS_estados_mensurados} by appending a weighted form of the linear
constraints \eqref{equality_constraint}.
Then, the optimal solution $\hat{{\theta}}_{\textrm{LS}}$ \eqref{CLS_estimator2} is approximated by the relaxed solution \citep{arablouei2015}
\begin{equation}\label{CLS_relaxed}
\hat{{\theta}}_{\rm{rCLS}}= \left(\Psi^{T}\Psi + \mu D^T D\right)^{-1}(\Psi^T Z + \mu D^T d),
\end{equation}
where $\mu \gg 1$ is the weight associated to the constraints.
If one tunes $\mu \to  \infty$, then $\hat{{\theta}}_{\rm{rCLS}} \to \hat{{\theta}}_{\rm{CLS}}$.
For applications in which the constraints \eqref{equality_constraint} are not precisely known, the relaxed solution \eqref{CLS_relaxed} is indicated. Indeed, other related least squares approaches may be used to solve this problem;  see \citet[Section 4.2]{2011_Teixeira_Aguirre}. 
\end{remark}

In \citet{Zhou_et_al_2001}, the recursive counterpart of \eqref{CLS_estimator2} is investigated. 
Interestingly,  \citet{Zhou_et_al_2001} shows that the CLS and the LS have the same recursive formulas, differing only at the initial values. In \citet{ZhuLi2007} it is shown  how to initialize the recursive least square equations in order to guarantee that the corresponding estimates satisfy \eqref{equality_constraint}, $\forall k$.
For mathematical convenience,  \citet{ZhuLi2007} and \citet{Zhou_et_al_2001} derive the recursive equations using Greville formulas, yielding equations in a non-standard format. 

Next, for simplicity, we present the recursive counterpart of \eqref{CLS_estimator2} in a more conventional format.  For $k = 1,\ldots, N$, we have 
\begin{eqnarray} 
\label{eq:rls_k} K_{k} &=& \frac{P_{{\textrm{CLS}},k-1}\psi_{k-1}}{\psi^{T}_{k-1}P_{{\textrm{CLS}},k-1}\psi_{k-1}+1},\\ 
\hat{\theta}_{{\textrm{CLS}},k}&=&\hat{\theta}_{{\textrm{CLS}},k-1}+K_{k}\left( z_{k}-\psi^{T}_{k-1}\hat{\theta}_{{\textrm{CLS}},k-1}\right), ~~~~\\ 
\label{eq:rls_P}  P_{{\textrm{CLS}},k}&=&\left(I_{n_p}-K_k\psi_{k-1}^T\right)P_{{\textrm{CLS}},k-1}.
\end{eqnarray} 
with initial values 
\begin{eqnarray}
\label{eq:rls_init1} \hat{\theta}_{{\textrm{CLS}},0} &=& \mathcal{P}_{\mathcal{N}(D)} \theta_{0}+\bar{d}, \\   
\label{eq:rls_init2} P_{{\textrm{CLS}},0} &=& \mathcal{P}_{\mathcal{N}(D)} P_{0},
\end{eqnarray}
where the projector  $\mathcal{P}_{\mathcal{N}(D)}$ \eqref{eq:projector}  guarantees that $\hat{\theta}_{{\textrm{CLS}},0}$ and $P_{{\textrm{CLS}},0}$ are compatible with \eqref{equality_constraint} for any ${\theta}_{0}$ and $P_{0}$.
We point out that \eqref{eq:rls_k}-\eqref{eq:rls_P} correspond to the classical {\em recursive least squares} (RLS).
The next result proves in a simple way that, if the RLS is properly initialized as in \eqref{eq:rls_init1}-\eqref{eq:rls_init2}, then its estimates satisfy  \eqref{equality_constraint}, $\forall k$. As mentioned above, a similar result is presented in \citet{ZhuLi2007} using Greville formulas. 
  
 \begin{proposition} \label{prop:rcls}
Assume that the initial parameter estimates $\hat{\theta}_{{\textrm{CLS}},0}$ and $P_{{\textrm{CLS}},0}$ are given by \eqref{eq:rls_init1}-\eqref{eq:rls_init2} and are used to initialize  the RLS equations \eqref{eq:rls_k}-\eqref{eq:rls_P}. Then the estimates given by \eqref{eq:rls_k}-\eqref{eq:rls_P} satisfy the time-invariant constraint $D\hat{\theta}_{{\rm CLS},k}=d, \forall k$.
\end{proposition}
\begin{proof}
See Appendix \ref{appendixA}.
\end{proof}

\subsection{Time-varying case}

Assume now that the parameters may vary with time such that \eqref{regression_eq} is replaced by
\begin{equation}\label{regression2_eq}
z_k =\psi^{T}_{k-1}\hat{\theta}_{k}+\xi_k.
\end{equation}
Also, assume that $\hat{\theta}_{k}$  satisfy the known time-varying constraint
\begin{equation}\label{equality_constraint2}
D_k\hat{\theta}_k=d_k.
\end{equation}

For $k = 1,\ldots, N$, the recursive time-varying counterpart of \eqref{CLS_estimator2} is given by \citet{Alenany2013} 
\begin{eqnarray} \label{eq:rls_ff_k}
K_{k} &=& \frac{P_{{\rm WLS},k-1}\psi_{k-1}}{\psi^{T}_{k-1}P_{{\rm WLS},k-1}\psi_{k-1}+\lambda},\\ 
\hat{\theta}_{{\rm WLS},k}&=&\hat{\theta}_{{\rm WLS},k-1}+K_{k}\left( z_{k}-\psi^{T}_{k-1}\hat{\theta}_{{\rm WLS},k-1}\right), ~~~~\\ 
P_{{\rm WLS},k}&=&\frac{1}{\lambda}\left(I_{n_p}-K_k\psi_{k-1}^T\right)P_{{\rm WLS},k-1}, \label{eq:rls_ff_P_k}\\\
L_k&=& P_{{\rm WLS},k}D_k^{T}[D_kP_{{\rm WLS},k}D_k^{T}]^{-1}, \label{eq:rls_ff_L_k}\\  \label{eq:rls_ff_theta_CRLS}
\hat{{\theta}}_{{\textrm{WCLS}},k}&=& (I_{n_p}-L_kD_k)\hat{{\theta}}_{{\rm WLS},k}+L_kd_k,
\end{eqnarray}
where $ 0 \ll \lambda \le 1$ is the forgetting factor. We refer to this method as the recursive weighted constrained LS (RWCLS).

\begin{remark}\label{rem:initialization_ff}
In \citet{Alenany2013}, part of the identification data is used to estimate offline the initial parameters $\hat{\theta}_{{\textrm{WLS}},0}$ and $P_{{\textrm{WLS}},0}$ for \eqref{eq:rls_ff_k}-\eqref{eq:rls_ff_theta_CRLS} by means of the constrained batch algorithm \eqref{CLS_estimator2}. 
Instead, we suggest to initialize \eqref{eq:rls_ff_k}-\eqref{eq:rls_ff_theta_CRLS} as in \eqref{eq:rls_init1}-\eqref{eq:rls_init2}.
\end{remark}

\section{Equality-Constrained Least Squares for State-Space Modeling}\label{sec:CLS_LEC}

In order to solve Problem \ref{problem1}, first it is necessary to map the constraint on the state vector $x_k$ given by \eqref{restricoes_igualdade} to a constraint on the parameter vector $\Theta$ given by \eqref{eq:Theta}. 

The next results address this point by indicating  conditions for a state-space model to 
have a state vector satisfying an equality constraint. 
\begin{lemma}\citep[Proposition 3.1]{Teixeira2009} \label{lemma:ec}
For the system given by \eqref{eq:proc}, assume that 
\begin{eqnarray}
\label{SF_0_iqualdade} SG&=&0_{n_{r} \times q},\\ 
\label{SA_S_iqualdade} SA&=&S\\ 
\label{SB_0_iqualdade} SB&=&0_{n_{r} \times p}.
\end{eqnarray}
Then, for all $k \geq1$, 
$Sx_{k}=s$, where $s=Sx_{0}$.
\end{lemma}
\vspace{-0.5cm}
\begin{lemma}\label{lemma:ec2}
Assume that the state vector of  system \eqref{eq:proc} satisfies the equality constraint $Sx_{k}= s,~ \text{$\forall k \geq 0$}$.
Then, the relations \eqref{SF_0_iqualdade}, \eqref{SA_S_iqualdade} \eqref{SB_0_iqualdade} are hold.
\end{lemma}

\begin{proof} Multiplying \eqref{eq:proc} by $S$, we obtain $Sx_{k}=SAx_{k-1}+SBu_{k-1}+SGw_{k-1}$. Then $Sx_{k}=s$ implies that  $SA=S$,  $SB=0_{n_{r} \times p}$ and $SG=0_{n_{r} \times q}$.
\end{proof}

Lemma \ref{lemma:ec} gives the conditions \eqref{SF_0_iqualdade}-\eqref{SB_0_iqualdade}  for the dynamic system to satisfy \eqref{restricoes_igualdade}, while Lemma \ref{lemma:ec2} proves the counterpart.
In other words, the previous results provide conditions for process model \eqref{eq:proc} to be {\em compatible} \citep{2016_Li} with the state equality constraint \eqref{restricoes_igualdade}.

\begin{remark}
In \citet[Proposition 3.2]{Teixeira2009}, it is proved that if \eqref{restricoes_igualdade} holds, then the system is not controllable in $\mathbb{R}^n$ from the process noise $w_k$, but it is rather controllable in the subspace defined by \eqref{restricoes_igualdade}.
Then, we can replace $G w_{k-1}$ in \eqref{eq:proc} by $\tilde{w}_{k-1} \triangleq G w_{k-1}$ with singular noise covariance $\tilde{Q} \triangleq G Q G^T$ providing \eqref{SF_0_iqualdade} is verified as in  Lemma \ref{lemma:ec}. In so doing, we focus on relations \eqref{SA_S_iqualdade}-\eqref{SB_0_iqualdade}, which are related to the matrices $A$ and $B$ to be estimated.
\end{remark}

From \eqref{SA_S_iqualdade}-\eqref{SB_0_iqualdade}, we obtain the equality constraint on the parameter matrix 
\begin{equation} \label{eq:ectheta}
\Theta  {D}_1=  {D}_2,
\end{equation}
where $\Theta$ is given by \eqref{eq:Theta} and
\begin{eqnarray}\label{constraint}
{D}_1 \triangleq S^T  ~{\rm and}~ {D}_2 \triangleq \left[\begin{array}{c}
         S^{T} \\
         0_{p\times n_{r} }
       \end{array}\right],
\end{eqnarray}
where $D_1 \in \mathbb{R}^{n \times n_{r}}$ and ${D}_2 \in \mathbb{R}^{(n+p) \times n_{r}}$. 

\begin{remark} \label{rem:mi}
Note that the estimator \eqref{CLS_estimator2} cannot be used to solve Problem \ref{problem1} due to matrix size incompatibility. In \eqref{CLS_estimator2}, the equality constraint \eqref{equality_constraint} is enforced on vector $\hat{\theta}$, whereas in Problem \ref{problem1}, the equality constraint \eqref{eq:ectheta} is enforced on matrix $\hat{\Theta}$, which is right-multiplied by $D_1$.
\end{remark}

The next result rewrites the matrix equations \eqref{eq:ss_lr} and \eqref{eq:ectheta} onto vectorized equations like \eqref{regression_eq} and \eqref{equality_constraint} such that a classical equality-constrained least squares problem is obtained. In so doing, we have a solution for Problem \ref{problem1}.
Likewise, based on this result, the recursive solution can also be obtained for both time-invariant and time-varying cases.

\begin{proposition} \label{prop:cls}
For the linear regression model \eqref{state_1} with parameter $\hat{\Theta}$ given by \eqref{eq:Theta}, 
the parameter estimate $\hat{\Theta}_{\rm CLS}$ minimizes $J_{\rm LS}(\hat{\Theta})$ given by \eqref{JLS_estados_mensurados_matrizes} subject to the equality constraint \eqref{eq:ectheta},  if and only if $\hat{\theta}_{\rm CLS}$ given by \eqref{CLS_estimator2} minimizes $J_{\rm LS}(\hat{\theta})$ given by \eqref{JLS_estados_mensurados} subject to the equality constraint \eqref{equality_constraint} with
\begin{eqnarray}
\label{eq:c1}
&\hat{\theta} = \left[ \begin{array}{c} 
\text{\rm{vec}}(A)\\
 \text{\rm{vec}}(B)
 \end{array} \right], ~
Z = \left[ \begin{array}{c}
\text{\rm{vec}}(y_{1})\\
\text{\rm{vec}}(y_{2})\\
 \vdots\\
\text{\rm{vec}}(y_{N})
 \end{array} \right],&  \\ 
 \label{eq:c2}
 &\Psi = \left[\begin{array}{cc}
                              y_{0}^{T}\otimes I_{n} &  u_{0}^{T}\otimes I_{n} \\
                              y_{1}^{T}\otimes I_{n} &  u_{1}^{T}\otimes I_{n} \\
                               \vdots                         &              \vdots              \\ 
                              y_{N-1}^{T}\otimes I_{n} &  u_{N-1}^{T}\otimes I_{n} 
                              \end{array}
\right],  & \\
\label{eq:c3}
 &\hspace{-0.4cm} D = \left[\begin{array}{cc}
                              I_{n} \otimes S				&  0_{n_{r} n \times np}  \\
                              0_{n_{r} p \times n^2}                      &  I_{p} \otimes S
                              \end{array}
\right] \hspace{-0.1cm},  d = \left[ \begin{array}{c} 
\text{\rm{vec}}(S)\\
\text{\rm{vec}}(0_{n_{r}  \times p})
 \end{array} \right],& \nonumber \\ &&
\end{eqnarray}
where $\theta \in \mathbb{R}^{(n^2+np)}$, $Z \in \mathbb{R}^{Nn}$, $\Psi \in \mathbb{R}^{Nn \times (n^2+np)}$, ${{D}} \in  \mathbb{R}^{(n_{r}n+n_{r}p)\times (n^2+np)}$, and ${{d}} \in \mathbb{R}^{(n_{r}n+n_{r}p)}$. 
\end{proposition}
\begin{proof}
This proof has two parts. First, we rewrite \eqref{eq:ss_lr} as \eqref{regression_eq}. This is done by using the following relation \citep{Bernstein2005}
\begin{equation}\label{vec_operator}
\textrm{vec}(\mathcal{M}\mathcal{N}\mathcal{O})=(\mathcal{O}^{T}\otimes \mathcal{M})\textrm{vec}(\mathcal{N}),
\end{equation}
where $\mathcal{M}$, $\mathcal{N}$ and $\mathcal{O}$ are real matrices of appropriate size.

The deterministic part of the model \eqref{state_1} can be rewritten as  
\begin{equation}
\text{\rm{vec}}(y_{k+1})=\left[\begin{array}{cc}
                              y_{k}^{T}\otimes I_{n} &  u_{k}^{T}\otimes I_{n} 
                              \end{array}
\right]  \left[ \begin{array}{c} 
\text{\rm{vec}}(A)\\
 \text{\rm{vec}}(B)
 \end{array} \right], \nonumber
\end{equation}
which is the form of \eqref{regression_eq}.

Second, we rewrite \eqref{eq:ectheta} as \eqref{equality_constraint}. For a data set of size $N$,  we have 
\begin{equation}
\left[ \begin{array}{c} 
\text{\rm{vec}}(y_{1})\\
\text{\rm{vec}}(y_{2})\\
 \vdots\\
\text{\rm{vec}}(y_{N})
 \end{array} \right] = \left[\begin{array}{cc}
                              y_{0}^{T}\otimes I_{n} &  u_{0}^{T}\otimes I_{n} \\
                              y_{1}^{T}\otimes I_{n} &  u_{1}^{T}\otimes I_{n} \\
                               \vdots                         &              \vdots              \\ 
                              y_{N-1}^{T}\otimes I_{n} &  u_{N-1}^{T}\otimes I_{n} 
                              \end{array}
\right] \left[ \begin{array}{c} 
\text{\rm{vec}}(A)\\
 \text{\rm{vec}}(B)
 \end{array} \right], \nonumber
\end{equation}
which is an equation of the form \eqref{regression_eq}, that is, $Z = \Psi\theta$, with $Z$ and $\Psi$ set as in \eqref{eq:c1}-\eqref{eq:c2}.

Thus,  the cost function $J_{\rm LS}(\hat{\theta})$ given by \eqref{JLS_estados_mensurados} is equivalent to the cost function $J_{\rm LS}(\hat{\Theta})$ given by \eqref{JLS_estados_mensurados_matrizes}.

Likewise, the equality constraint \eqref{eq:ectheta} can be rewritten as
\begin{equation}
\left[\begin{array}{cc}
                      I_{n} \otimes S				&  0_{n_{r} n \times np}  \\
                              0_{n_{r} p \times n^2}                      &  I_{p} \otimes S
                              \end{array}
\right] \left[ \begin{array}{c} 
\text{\rm{vec}}(A)\\
\text{\rm{vec}}(B)
 \end{array} \right] 
=\left[ \begin{array}{c} 
\text{\rm{vec}}(S)\\
\text{\rm{vec}}(0_{n_{r} \times p})
 \end{array} \right], \nonumber
\end{equation}
which is an equation of the form \eqref{equality_constraint}, with $D$ and $d$ set as in \eqref{eq:c3}.
\end{proof}

\begin{corollary} \label{cor:cls}
Given the vectorizing definitions \eqref{eq:c1}-\eqref{eq:c3},
the recursive equality-constrained parameter estimates $\hat{\Theta}_{{\rm CLS},k}$ are given by \eqref{eq:rls_k}-\eqref{eq:rls_init2}. 
\end{corollary} 

\begin{proof}
Given the result of Proposition \ref{prop:cls} and that  the recursive equations \eqref{eq:rls_k}-\eqref{eq:rls_init2} are equivalent to the equality constrained least squares  \eqref{CLS_estimator2} (see Proposition \ref{prop:rcls}), we prove this result.  
\end{proof}

\begin{remark} For the time-varying counterpart of \eqref{eq:proc}-\eqref{restricoes_igualdade}, using definitions similar to \eqref{eq:c1}-\eqref{eq:c3}, the recursive equality-constrained parameter estimates $\hat{\Theta}_{{\rm CLS},k}$ can be obtained from \eqref{eq:rls_ff_k}-\eqref{eq:rls_ff_theta_CRLS}.
\end{remark}

Note that \eqref{eq:proc}-\eqref{eq:obs} characterizes an output-error model. So, the next result proves that the LS estimator is biased for such type of model. \begin{proposition} \label{prop:els}
For the state-space model \eqref{eq:proc}-\eqref{eq:obs}, the LS estimator given by \eqref{LS_estimator} and \eqref{eq:c1}-\eqref{eq:c2} is biased, that is, $\rm{E}[\hat{\theta} ] - \theta \neq 0$.  
 \end{proposition}
\begin{proof}
See Appendix \ref{appendixB}.
\end{proof}
Thus, in this work, we use algorithms based on the extended LS \citep{Ljung1987}; however, for brevity, we omit the term ``extended''. Other unbiased estimators could be used instead.
\vspace{-0.4cm}
\section{Simulated Results}\label{sec:SR}
\subsection{Compartmental system: time-invariant case}
Consider the linear discrete-time compartmental model  \citep{Teixeira2009} 
represented by \eqref{eq:proc}-\eqref{eq:obs} involving mass exchange among compartments whose matrices are given by

\begin{equation}\label{comp_syst_1}
A=\left[\begin{array}{ccc}
                              0.94 &  0.028 & 0.019 \\
                              0.038 &  0.95 & 0.001 \\
                               0.022 &  0.022 & 0.98 
                              \end{array}\right]; ~B= 0_{3 \times 1}; ~~ C=I_{3 \times 3}; 
                               \end{equation}
with state vector $x_k$ $\in$ $\mathbb{R}^3$ composed by the amount of mass in each compartment, initial condition $x_{0}$ $=$ $[1~~ 1~~ 1]^{T}$ , and process noise and observation noise covariance matrices \hbox{$\tilde{Q}$ $=$ $\sigma_{w}^2 G G^T$}, where $G=\left[\begin{array}{cc}
                              0.05 &  -0.03  \\
                              -0.02 & 0.01  \\
                              -0.03 & 0.02  
                              \end{array}\right]$, and $R_{k} = \sigma_{v}^2 I_{2 \times 2}$. 
                              
\vspace{0.2cm}                             
                              One realization of simulated identification data for this system is shown in Fig. \ref{fig:3dstates} for $\sigma_{w}$ $=$ $1.0$ and  \hbox{$\sigma_{v}$ $=$ $0.1$}. 
 Note that conditions of  Lemmas \ref{lemma:ec} and \ref{lemma:ec2} hold for \eqref{comp_syst_1} such that the trajectory of $x_{k}$ $\in$ $\mathbb{R}^{3}$ lies on the plane \eqref{restricoes_igualdade}, whose parameters are assumed to be known and are given by
 \begin{equation}\label{eq:compartmental_constraints}
 S=[1~~1~~1],~~s=3,
 \end{equation}
 that is, mass conservation  is verified. The validation data is simulated with different initial condition $x_{0}$ $=$ $[2~~1 ~~ 0]^{T}$. 

  \begin{figure}[]
  \centering
    \includegraphics[scale=0.75]{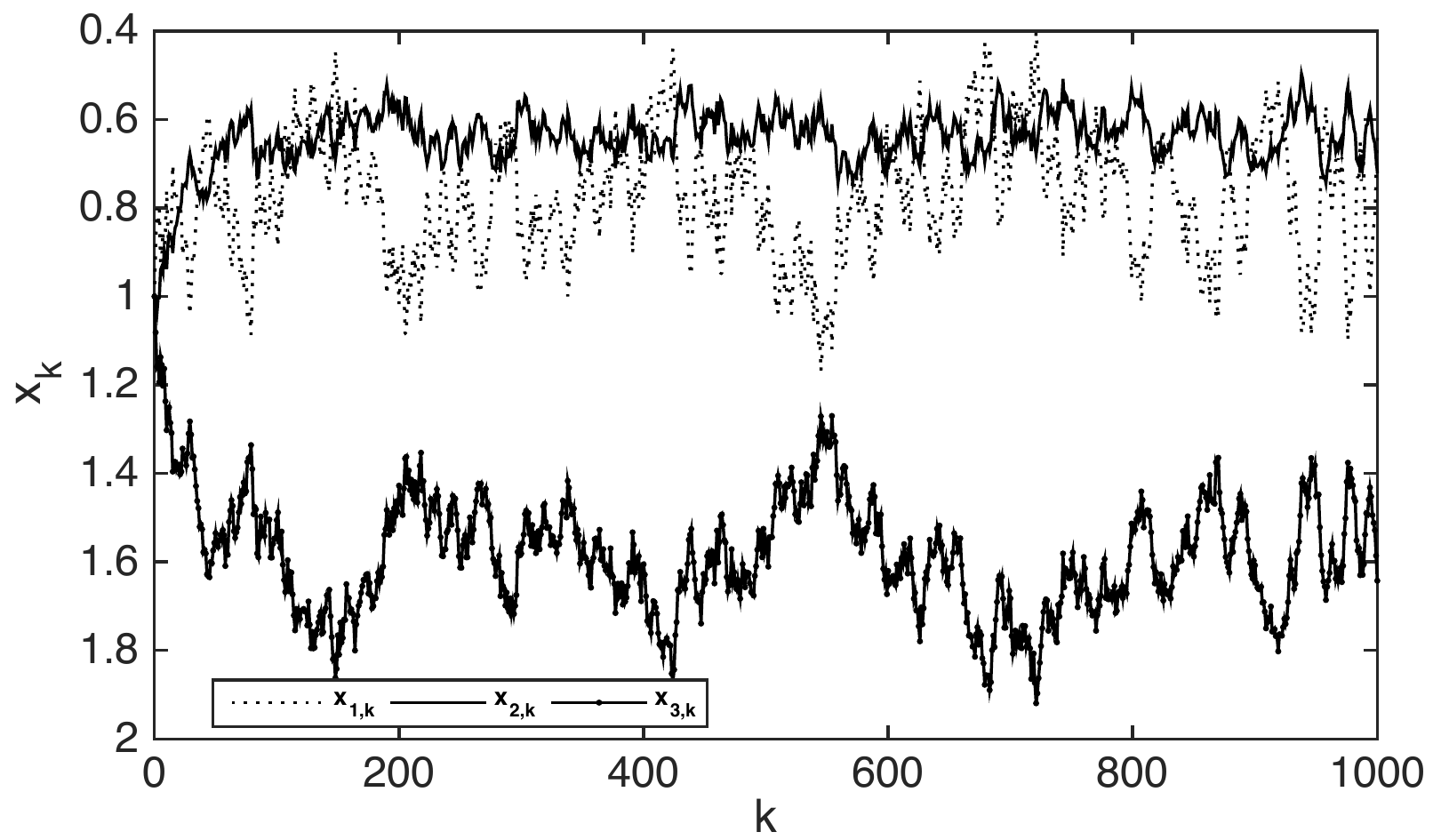}
      \caption{One realization of identification data from the compartmental model. The state components are shown evolving with time.}
      \label{fig:3dstates}
\end{figure}

We investigate a 1000-run Monte Carlo simulation testing identification data with different noise realizations for $\sigma_{w}$ $=$ $1.0$ and  \hbox{$\sigma_{v}$ $=$ $0.1$}. The LS given by \eqref{LS_estimator} for the state-space system, for which \eqref{eq:ss_lr}-\eqref{JLS_estados_mensurados_matrizes} are defined, is used to yield estimates for the matrix $A$ as discussed in Section \ref{sec:ps}. Likewise, as indicated by Proposition \ref{prop:cls}, CLS was also employed. 

Fig. \ref{fig:freerun_compart_system} shows the results regarding the Monte Carlo validation of the obtained models (mean values with two standard-devation confidence interval). In order to quantify the fit between the simulation of the system and the identified models, we use the root-mean-square error  for each  $i$th state component, $i=1,\ldots,n$,  
\begin{equation}
\textrm{RMSE}^{m_{r}}_{n}\triangleq\sqrt{\frac{\sum^{N}_{k=1}(x_{i,k}-\hat{x}_{i,k})^{2}}{N}},
\end{equation}
where  $N$ is the length of the measured data and \hbox{$m_{r}=1,\ldots,M$}, where $M$ is the number of realizations. Table \ref{table1} shows the mean $\overline{\textrm{RMSE}}_{n}$ and standard deviation $\sigma_{\rm{RMSE}_{n}}$  of the RMSE for each state.
Note that the performance of the model obtained with CLS is better than the model estimated by LS. 
That is, the auxiliary information about mass conservation was useful.  

\begin{figure}[]
  \centering
    \includegraphics[scale=0.75]{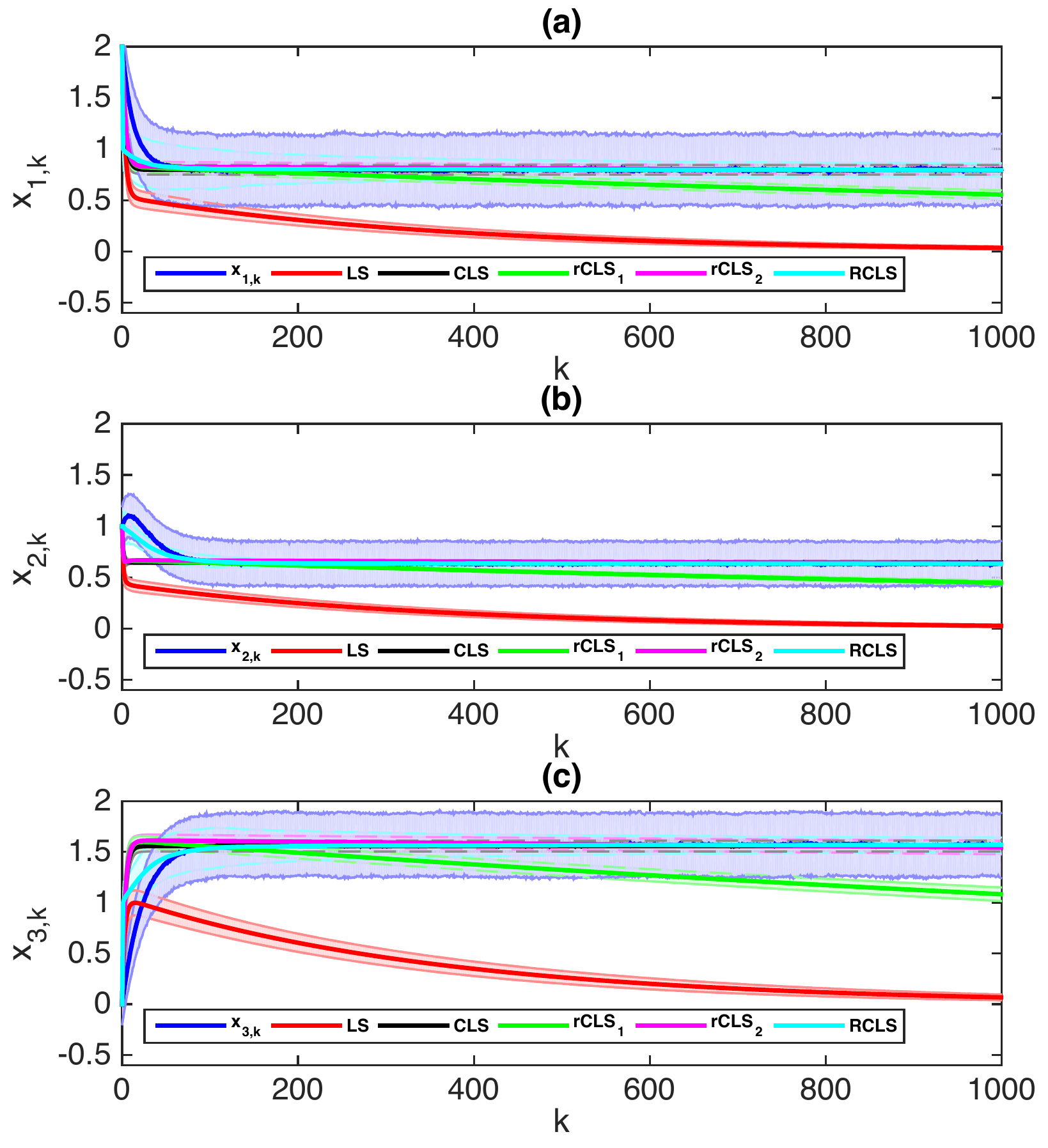}
    \caption{Comparison of $1000$-run Monte Carlo simulations for the validation of   state-space models identified using LS (red line), CLS (black line) and RCLS (cyan line). To address the case of uncertain prior information, we estimate models using rCLS, for which we test two values of the tuning parameter: $\mu_{1}=5\times10^{3}$ (green line) and $\mu_{2}=5\times10^{4}$ (magenta line). The mean of true values of the states are shown in a blue continuous line. In all cases, the mean values of the validation data are shown within the confidence interval of two standard deviations given by the respective light shaded  areas. \label{fig:freerun_compart_system}}
\end{figure}

In addition, we consider the case where the auxiliary information is uncertain. Suppose that the uncertain state equality constraint is assumed to be given by  \eqref{restricoes_igualdade} with
\begin{equation}\label{eq:compartmental_constraints_uncertain}
 S =[1.4~~0.9~~1.2],~~s =3.5. 
 \end{equation}
 The rCLS given by \eqref{CLS_relaxed} is used to estimate the state-space model with the uncertain auxiliary information \eqref{eq:compartmental_constraints_uncertain}. As in \citet{arablouei2015}, we tuned the parameter $\mu$ in order to obtain models with good prediction performance.  The results are shown in \linebreak Fig. \ref{fig:freerun_compart_system}  and  Table \ref{table1} for $\mu_{1}=5\times10^{3}$  ($\rm{rCLS}_{1}$) and $\mu_{2}=5\times10^{4}$ ($\rm{rCLS}_{2}$). Note that the results yielded by ($\rm{rCLS}_{1}$) are better than those from LS. Moreover,  results from $\rm{rCLS}_{2}$ almost coincide to those from  CLS. Then, the appropriate use of  uncertain prior information may improve the quality of the estimated model, as discussed in  \citet{2011_Teixeira_Aguirre}.

We also test the recursive solution to this problem as indicated by Corollary \ref{cor:cls}.
 RLS is properly initialized as in \eqref{eq:rls_init1}-\eqref{eq:rls_init2}, with $\theta_{0}$ $=$ $0_{9 \times 1}$ and $P_{0}$ $=$ $10^3I_{9 \times 9}$, yielding RCLS and is compared with the batch LS and CLS estimates in Fig. \ref{fig:freerun_compart_system} and  Table \ref{table1}.  Note that, when compared to LS, the use of auxiliary information \eqref{eq:compartmental_constraints} in the initialization of  RCLS improves the performance of the estimated model. 
\begin{table}
\tbl{The mean and the standard deviation of the RMSE for 1000-run Monte Carlo simulations of each state sequence.}
{\begin{tabular}{p{34pt}p{67pt}p{67pt}p{67pt}}
\hline
Method & $\hfil \! \overline{\textrm{RMSE}}_{1}\!\! \pm \!\sigma_{\rm{RMSE}_{1}}$ & $\hfil \! \overline{\textrm{RMSE}}_{2}\!\! \pm \!\sigma_{\rm{RMSE}_{2}}$ & $\hfil \! \overline{\textrm{RMSE}}_{3}\!\! \pm \!\sigma_{\rm{RMSE}_{3}}$    \\ \hline
LS &   			\hfil 0.660   $\pm$  0.029   & \hfil 0.527 $\pm$ 0.013 & \hfil 1.248 $\pm$ 0.038 \\ 
CLS &  			\hfil 0.191   $\pm $ 0.012   & \hfil 0.134 $\pm$ 0.003 & \hfil 0.211 $\pm$ 0.014 \\ 
$\rm{rCLS_{1}}$ &    \hfil 0.232   $\pm$  0.020   & \hfil 0.169 $\pm$ 0.007 & \hfil 0.356 $\pm$ 0.031 \\ 
$\rm{rCLS_{2}}$ &    \hfil 0.189   $\pm$  0.012  &  \hfil 0.133 $\pm$  0.003& \hfil 0.219 $\pm$ 0.016 \\ 
RCLS & 			 \hfil 0.202  $\pm$  0.013   & \hfil 0.116 $\pm$ 0.003&  \hfil 0.196 $\pm$ 0.015\\  \hline
\end{tabular}}
\label{table1}
\end{table}
\vspace{-0.5cm}

\subsection{Compartmental system: time-varying case}
We now consider a time-varying compartmental system.
As in a reconfigurable system, we consider the case in which the linear dynamics switches among  three different modes. 
For instance, this may be the case for a multi-tank system with reconfigurable valves.
The first mode is simulated with $A_1$ as in \eqref{comp_syst_1}. The second and third modes are described by the matrices 
 $$A_{2}=\left[\begin{array}{ccc} 
                              0.84 &  0.028 & 0.019 \\
                              0.138 &  0.85 & 0.001 \\
                               0.022 &  0.122 & 0.98 
                              \end{array}\right],$$ 
$$A_{3}=\left[\begin{array}{ccc} 
                              0.80 &  0.018 & 0.119 \\
                              0.178 &  0.76 & 0.201 \\
                               0.022 &  0.222 & 0.68 
                              \end{array}\right].$$                       
The matrices $B$ and $C$ are defined as in \eqref{comp_syst_1} for all modes.

\begin{figure}[]
  \centering
    \includegraphics[scale=0.75]{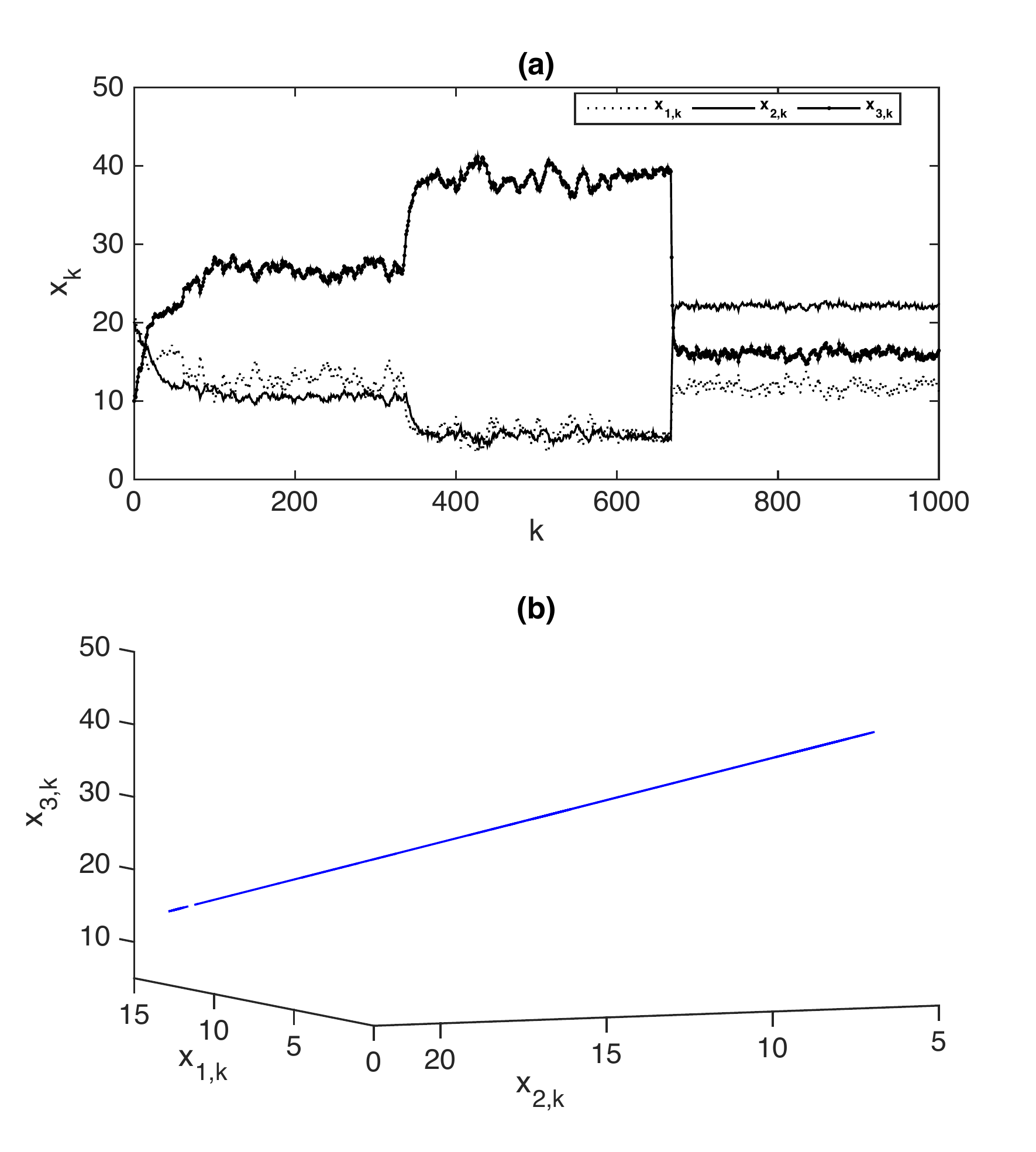}
      \caption{One realization of the identification data for the time-varying compartmental model. In (a), the state components are shown evolving with time and, in (b), in state space. Note that, regardless of the mode, for all $k$ $\geq$ $0$, $x_k$ $\in$ $\mathbb{R}$ lies on the plane $x_{1,k}$ $+$ $x_{2,k}$ $+$ $x_{3,k}$ $=$ $50$.} 
      \label{fig:3dstates_3OP}
\end{figure}

A typical realization of the simulated identification data is shown in Fig. \ref{fig:3dstates_3OP} for $\sigma_{w}$ $=$ $10$ and  \hbox{$\sigma_{v}$ $=$ $1$} and \hbox{$x_{0}$ $=$ $[20~~20~~10]^{T}$}. The mass conservation is verified for all operating points and the assumedly known parameters of   \eqref{restricoes_igualdade} are given by 
\begin{equation}\label{eq:compartmental_constraints_2}
 S=[1~~1~~1],~~s=50.
 \end{equation} 
 Note that the conditions of  Lemmas \ref{lemma:ec} and \ref{lemma:ec2} are verified for all modes. 
 A new initial condition is set as \hbox{$x_{0}$ $=$ $[15~~10~~25]^{T}$} to simulate the validation data. 

We generate $1000$ Monte Carlo simulations with different noise realizations for $\sigma_{w}$ $=$ $10$ and  \hbox{$\sigma_{v}$ $=$ $1$} in order to obtain the identification data. For each running simulation we employ both RWLS given by \eqref{eq:rls_ff_k}-\eqref{eq:rls_ff_P_k} and RWCLS given by \eqref{eq:rls_ff_k}-\eqref{eq:rls_ff_theta_CRLS} to estimate the time-varying model.  These recursive estimators are randomly initialized with an arbitrary initial condition  ${\rm vec}(\Theta_{0}) \in \mathbb{R}^{{9 \times 1}}$ given by a normal distribution with $\sigma_{\Theta}=1$ and $P_{0}=10^{4}I_{9 \times 9}$. Recall that in \citet{Alenany2013}  part of the identification data is used to estimate de initial conditions by means of the batch algorithm \eqref{CLS_estimator2}; see Remark \ref{rem:initialization_ff}. Here, we use the result given by \eqref{eq:rls_init1}-\eqref{eq:rls_init2} to more conveniently proceed the identification procedure using the RWCLS. 
The forgetting factor $\lambda$  is set to $0.95$. 

The Monte Carlo validation results for RWLS and RWCLS  are shown in Fig. \ref{fig:freerun_compart_system_RLS_ff2}. We would like to draw attention to the variance of the estimated models. For the three different modes, we verify that the performance of the model obtained with RWCLS is better than the model estimated by RWLS. That is, the prior information about mass conservation improved the quality of the estimated model.  
 \vspace{-0.3cm}
 
 \begin{figure}[]
  \centering
    \includegraphics[scale=0.76]{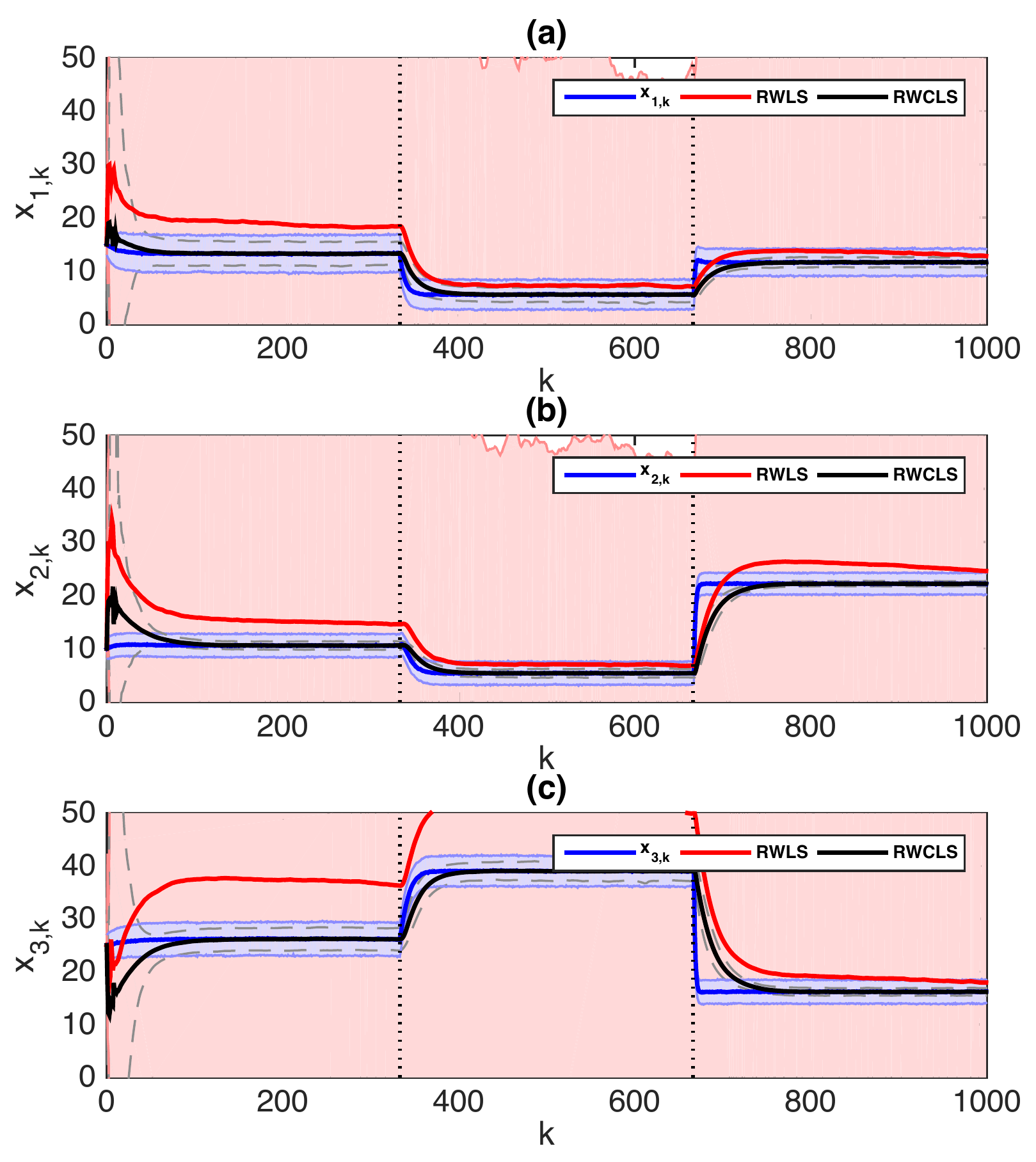}
\caption{Comparison of $1000$-run Monte Carlo simulations for the validation of   state-space models identified using RWLS (red line) and RWCLS (black line), where the RWCLS is initialized in according to \eqref{eq:rls_init1}-\eqref{eq:rls_init2} and the constraint is enforced with \eqref{eq:rls_ff_L_k}-\eqref{eq:rls_ff_theta_CRLS}. The mean of true values of the states are shown in a blue continuous line. In all cases, the mean values of the validation data are shown within the confidence interval of two standard deviations given by the respective light shaded  areas. Observe that the confidence interval of RWLS are partially presented. Dotted vertical black lines show when the operating points are switched.}      \label{fig:freerun_compart_system_RLS_ff2}

\end{figure}

 \subsection{Practical application: forest ecosystem}
Consider the compartmental model of nitrogen flow in a tropical forest studied in \citet{Walter1999} and summarized in Figure \ref{fig:figura_artigo_forest_ecosystem}, where $k_{ij}$ accounts for the flow rates between compartments, associated with the mass leaving the $i$th compartment and arriving at the $j$th compartment. Note that each compartment allows interaction from its neighbour compartments in both directions and that the parameters are given for a continuous-time model. Choosing the sampling-period as $T_{s}=0.1$ years, the discrete-time matrices are given by
\begin{eqnarray}\label{comp_syst_3}
A\!\!&\!\!=\!\!&\!\!\left[\begin{array}{ccccc} \nonumber
                              0.9003 &  0          & 0.0005 & 0         &  0.0093 \\
                              0.0935 &  0.8807  & 0         & 0         & 0.0005 \\
                              0.0054 &  0.0978 & 0.6697 & 0          & 0 \\
                              0.0005 &  0.0154 & 0.2372 & 0.9995 & 0 \\
                              0.0002 &  0.0060 & 0.0927 & 0.005  &0.9902
                              \end{array}\right], \\
                              B\!\!&\!\!=\!\!&\!\! \left[\begin{array}{ccccc} \nonumber
                              0.5505\!&\!0.0282\!&\!-0.2625\!&\!-0.3003\!&\! -0.0159
                              \end{array}\right]^{T},\\
                              C\!\!&\!\!=\!\!&\!\!I_{5 \times 5},  \label{continuous_model_ecosystem} \end{eqnarray}
with state vector $x_k$ $\in$ $\mathbb{R}^5$ composed by the amount of nitrogen in each compartment, initial condition \hbox{$x_{0}\!\!=\!\![\!-\!3.5~\!-\!2.52~~0~~520~~26.5]^{T}\!\!+\!\!\mathcal{K}[3.82~~316~~1~~576~~41\!]^{T}$}, where $\mathcal{K}$ depends on the total amount of nitrogen in the system in the beginning and is set as $\mathcal{K}=1.5$, 
 the process noise and observation noise covariance matrices \hbox{$\tilde{Q}$ $=$ $\sigma_{w}^2 G G^T$}, where $G=\left[\begin{array}{cccc}
                              0.1220 &  0.1634 &0.0249 &  -0.0383  \\
                              -0.0420   &-0.0048  &-0.1430  & 0.0235  \\
                               0.1640 & -0.0317  &-0.0057 &  0.0571 \\
                                  -0.1871 &  -0.0877  &0.1697  & -0.0098  \\
                                   -0.0569 & -0.0392  &-0.0459  &  -0.0325  
                              \end{array}\right]$, and $R_{k} = \sigma_{v}^2 I_{4 \times 4}$. 
The conditions of  Lemmas \ref{lemma:ec} and \ref{lemma:ec2} hold for \eqref{comp_syst_3} such that the trajectory of $x_{k}$ $\in$ $\mathbb{R}^{5}$ lies on the plane \eqref{restricoes_igualdade}, whose parameters are assumed to be known and are given by
 \begin{equation}\label{eq:compartmental_constraints_ex_3}
 S=[1~~1~~1~~1~~1],~~s=1.9472 \times 10^{3}.
 \end{equation}
 that is, mass conservation  is verified.
 
We investigate a 1000-run Monte Carlo simulation testing identification data with $u_{k}=1+\sigma_{u}w_{k}^u$ and different noise realizations for $\sigma_{u}$ $=$ $0.1$, $\sigma_{w}$ $=$ $1.0$ and  \hbox{$\sigma_{v}$ $=$ $1.0$} (not shown for brevity), where $w^{u}$ is a zero mean white noise to ensure the persistence of excitation of the input. For each running simulation we employ both LS given by  \eqref{LS_estimator}  and CLS given by Proposition \ref{prop:cls}. The validation data is
simulated with different initial condition \hbox{$x_{0}\!\!=\!\![72.2~~381.5~~101.5~~1264.0~~128.0]^{T}$}. Figure \ref{fig:histogram_compart_system_CLS} shows the results regarding the RMSE of the Monte Carlo validation of the obtained models. For all the states the use of auxiliary information in the CLS improves the performance of the estimated models. Specifically, the CLS estimation of the $4$th and $5$th state components always outperform the LS estimation. For in the CLS estimator, $\overline{\textrm{RMSE}}_{4}$  is around $15$ times smaller than the same index for LS. For the  $5$th state,  CLS improves $\overline{\textrm{RMSE}}_{5}$ by a factor of 3.
 \begin{figure}[]
  \centering
    \includegraphics[scale=0.6]{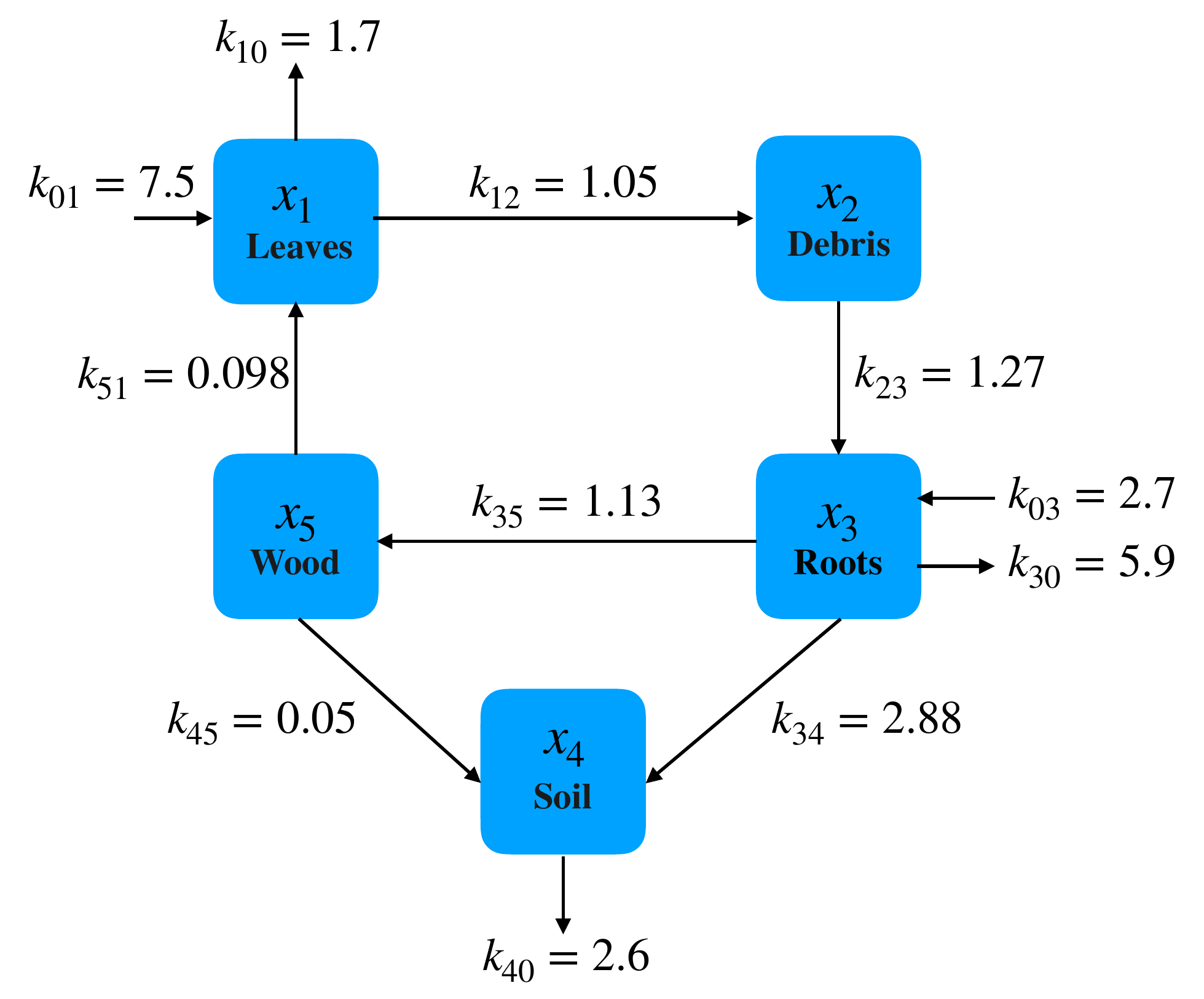}
\caption{Schematic compartmental model of nitrogen flow in a tropical forest. The inputs and outputs are assumed to be constant in the model. Adapted from \citet{Walter1999}. }      \label{fig:figura_artigo_forest_ecosystem}
\end{figure}
 \begin{figure}[]
  \centering
    \includegraphics[scale=0.85]{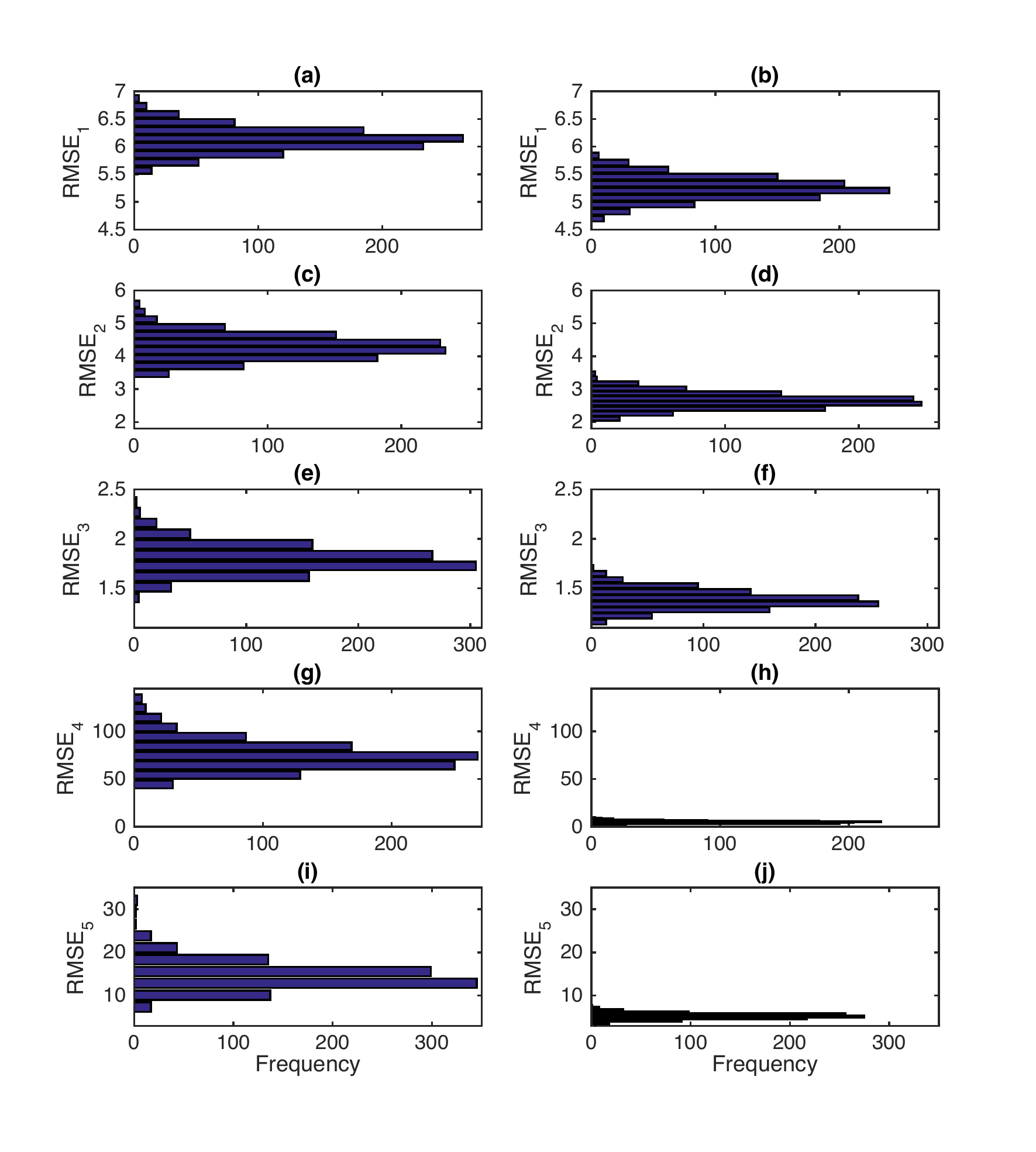}
\caption{Histogram of the RMSE between the validation data and LS (left) and CLS (right) estimation for all the $1000$-run Monte Carlo simulations.}      \label{fig:histogram_compart_system_CLS}
\end{figure}
 \vspace{-0.5cm}
\section{Concluding Remarks}\label{sec:conclusions}
We address the problem of modeling state-space dynamic systems for which the state vector satisfies an exactly known or an uncertain equality constraint (prior information). We assume that all state components are measured such that least-square methods can be used to estimate the state-space matrices. Both batch and recursive algorithms are considered. By means of the latter, the time-varying case is also addressed. 

First, we show how to map the known equality constraint on the state vector on an equality constraint on the parameter matrix to be estimated by the least-square based method. Then, we show how to rewrite the corresponding least squares problem into a vectorized form such that existing equality-constrained least squares methods may be used. 

In addition to obtaining state-space matrices that yield an equality-constrained model on the state vector, we observe that the usage of both exactly known  and uncertain prior information improves the prediction quality of the model compared to the case in which the equality constraint is not enforced. Such results are consistent with those from \citet{2011_Teixeira_Aguirre}. The algorithms here investigated are also of interest for gray-box subspace identification methods that employ least squares as an internal step; see \citet{Trnka2009,Alenany2011,privara2012,Alenany2013}.

\bibliographystyle{apacite}

\appendix
\section{} 
\label{appendixA}
Next, we present the proof of Proposition \ref{prop:rcls}.
\begin{proof} 
Given that \eqref{eq:projector} is a projector onto the null space of $D$, 
initialize the RLS equations given by \eqref{eq:rls_k}-\eqref{eq:rls_P} with the initial values \eqref{eq:rls_init1}-\eqref{eq:rls_init2} such that
\begin{eqnarray}
D \hat{\theta}_{\textrm{CLS},0} &=& d, \\
D P_{\textrm{CLS},0} &=&0.
\end{eqnarray}

For $k = 1$, multiplying \eqref{eq:rls_k}-\eqref{eq:rls_P} by $D$  we have 
\begin{eqnarray} \label{eq:DK_1}
DK_{1} &=& \frac{DP_{\textrm{CLS},0}\psi_{0}}{\psi^{T}_{0}P_{\textrm{CLS},0}\psi_{0}+1}=0,\\ 
\nonumber D\hat{\theta}_{\textrm{CLS},1}&=&D{\theta}_{\textrm{CLS},0}+DK_{1}\left( z_{1}-\psi^{T}_{0}{\theta}_{\textrm{CLS},0}\right),\\  \label{eq:DP_1}
&=&d,\\
\nonumber DP_{\textrm{CLS},1}&=&D\left(I_{n_p}-K_1\psi_{0}^T\right)P_{\textrm{CLS},0}, \\ 
 &=&DP_{\textrm{CLS},0}- DK_1\psi_{0}^TP_{\textrm{CLS},0}=0.
\end{eqnarray} 

By symmetry, from \eqref{eq:DK_1}-\eqref{eq:DP_1}, we verify that $DK_{2} =0$, $D\hat{\theta}_{\textrm{CLS},2}=d$ and $DP_{\textrm{CLS},2}=0$.

Likewise, at time $k+1$, we have $DK_{k+1}=0$, $D\hat{\theta}_{\textrm{CLS},k+1}=d$ and $DP_{\textrm{CLS},k+1}=0$.

Thus, by induction, we have $D K_{k}=0$, $D \hat{\theta}_{\textrm{CLS}, k}=d$ and $DP_{\textrm{CLS},k}=0$, \hbox{$\forall$ $k$}, completing the proof.
\end{proof}
\vspace{-0.32cm}
\section{}
 \label{appendixB}
In the following, we show that the classical LS estimator is biased for the estimation of the matrices of the state-space model \eqref{eq:proc}-\eqref{eq:obs}.
\begin{proof} 
Replace $y_{k}$ in \eqref{state_1}, write this result as the linear regression model \eqref{regression_eq} and apply the \textit{vectorization operador} \eqref{vec_operator}, then we obtain 
\begin{equation}
\textrm{vec}(y_{k+1})=\Psi_{k}^{T}\theta+\textrm{vec}(e_{k+1}). \nonumber
\end{equation}
where
 \begin{eqnarray} \label{state_expanded}
\Psi_{k}^{T}\! &\!\!\!=\!\!\!&\! \left[\!\!\!\begin{array}{cc}
         \left\{\!\left[\!\begin{array}{cc}
         y_{k-1}^{T} \! \!\otimes\! \!I_{n}& u^{T}_{k-1}\!\!\otimes \!\!I_{n}
       \end{array}\!\!\right]\theta+ \textrm{vec}(e_{k})\!\right\}\!\!\otimes\!\! I_{n} & u^{T}_{k} \!\!\otimes\! \! I_{n}
       \end{array}\!\!\right],\nonumber\\
   e_{k+1}\!&\!\!\!=\!\!\!&\! v_{k+1} - A\underline{v_k}+ G w_{k},\nonumber \\
e_{k} \!&\! \!\!=\!\!\!& \!\underline{v_{k}} - Av_{k-1}+ G w_{k-1}.\nonumber
\end{eqnarray}
The terms $v_{k}$ underlined are identical. So, the regressor $y_{k}$ between the braces is correlated with $e_{k+1}$. In this case, the LS estimator is biased. 
\end{proof}

\end{document}